\documentclass[reprint,amsmath,amssymb,aps]{revtex4-1}
\usepackage{algorithm}
\usepackage{algpseudocode}
\usepackage{amsthm}
\usepackage{amsfonts}
\usepackage{bbm}
\usepackage{graphicx}
\usepackage{epstopdf}
\usepackage{dcolumn}
\usepackage{bm}
\usepackage{braket}
\usepackage{color}
\usepackage{enumerate}
\usepackage{mathtools}
\usepackage{makecell}

\theoremstyle{plain}\newtheorem{theorem}{Theorem}
\theoremstyle{definition}
\theoremstyle{definition}
\theoremstyle{plain}\newtheorem{Pp}[theorem]{Proposition}
\theoremstyle{plain}
\theoremstyle{plain}
\theoremstyle{plain}
\newtheorem{defi}{Definition}
\begin{document}

\preprint{APS/123-QED}
\title{Exact distributed quantum algorithm for generalized Simon's problem}

\author{Hao Li, Daowen Qiu$^\dag$}
\email{$\dag$ issqdw@mail.sysu.edu.cn (D.W. Qiu, Corresponding author's address)}
\affiliation{
 Institute of Quantum Computing and Computer Theory, School of Computer Science and Engineering, Sun Yat-sen University, Guangzhou 510006, China; \\
 The Guangdong Key Laboratory of Information Security Technology, Sun Yat-sen University, 510006, China;\\
QUDOOR Technologies Inc., Zhuhai, China}

\author{Le Luo}
\affiliation{
	School of Physics and Astronomy, Sun Yat-sen University, 519082 Zhuhai, China;\\
QUDOOR Technologies Inc., Zhuhai, China
}
\author{Paulo Mateus}
\affiliation{
	Instituto de Telecomunica\c{c}\~{o}es, Departamento de Matem\'{a}tica,
	Instituto Superior T\'{e}cnico,  Av. Rovisco Pais 1049-001  Lisbon, Portugal
}
\date{\today}

\begin{abstract}
Simon’s problem is one of the most important problems demonstrating the power of quantum algorithms, as it greatly inspired the  proposal of Shor's algorithm. The generalized Simon's problem is a natural extension of Simon’s problem, and also a special hidden subgroup problem: Given a function $f:{\{0, 1\}}^n \to {\{0, 1\}}^m$, 
			with the property that for any $x, y\in {\{0, 1\}}^n$, there is some unknown hidden subgroup
			$S\leq\mathbb{Z}_2^n$ such that $f(x)=f(y)$ iff $x \oplus y\in S$, where $|S|=2^k$ for some $0\leq k\leq n$ $(m\geq n-k)$. 
			The goal of generalized Simon’s problem is to find the hidden subgroup $S$. In this paper, we present two key contributions. Firstly, we characterize the structure of  the generalized Simon's problem in distributed scenario and introduce a corresponding distributed quantum algorithm. Secondly, we refine the algorithm to ensure exactness due to the application of  quantum amplitude amplification technique. Our algorithm offers exponential acceleration compared to the distributed classical algorithm. When contrasted with the centralized quantum algorithm for  the generalized Simon's problem, our algorithm's oracle requires fewer qubits, thus making it easier to be physically implemented.  Particularly, the  exact distributed  quantum algorithm we develop for  the generalized Simon's problem outperforms the best previously proposed distributed quantum algorithm for Simon's problem in terms of generalizability and exactness.
\end{abstract}

\pacs{Valid PACS appear here}
\maketitle

\section{INTRODUCTION}{\label{Sec1}}

Quantum computing \cite{nielsen_quantum_2010} has been proved to have great potential in factorizing  large numbers  \cite{shor_polynomial-time_1997},  searching unordered database \cite{grover_fast_1996} and solving linear systems of equations \cite{HHL_2009}. However,  large-scale universal quantum computers have not yet been realized due to the limitations of current physical devices. At present, quantum technology has been entered to the Noisy
Intermediate-Scale Quantum (NISQ) era \cite{preskill_quantum_2018}, which makes it possible to implement quantum algorithms on middle-scale circuits.

Distributed quantum computing is a novel computing architecture, which combines quantum computing with distributed computing \cite{goos_distributed_2003,beals_efficient_2013,Qiu2017DQC,
caleffi_quantum_2018,avron_quantum_2021,
Qiu22,Tan2022DQCSimon,Xiao2023DQAShor,Hao2023DDJ,Xiao2023DQAkShor}. In  distributed quantum computing architecture, multiple quantum computing nodes communicate with each other  and cooperate to complete computing tasks. Compared with centralized quantum computing, the size and depth of circuit can be reduced by using distributed quantum computing, which is beneficial to improve the performance of  circuit against noise. 

Simon's problem is one of the most important problems in  quantum computing \cite{simon_power_1997}. For solving Simon's problem, quantum algorithms have the advantage of exponential acceleration over  classical algorithms \cite{cai_optimal_2018}. Remarkably,   Simon's algorithm greatly inspired the  proposal of Shor's algorithm \cite{shor_polynomial-time_1997}. Furthermore, the generalized Simon's problem is a natural extension of Simon’s problem,  and an instance of the hidden subgroup problem \cite{nielsen_quantum_2010,kaye_introduction_2007}.

Tan, Xiao, and Qiu et al. \cite{Tan2022DQCSimon}  proposed a distributed quantum algorithm for Simon's problem, but they left it open as to whether an exact distributed version exists. In the centralized case, Cai and Qiu \cite{cai_optimal_2018} utilized quantum amplitude amplification to address the issue of exactness for Simon's problem. Their approach has served as inspiration for our work. In this paper, we contribute in two new ways. Firstly, we characterize the structure of  the generalized Simon's problem in  distributed scenario and leverage this understanding to design a corresponding distributed quantum algorithm. Secondly, we incorporate quantum amplitude amplification \cite{BHMT02} to ensure the algorithm's exactness.

The remainder of this paper is organized as follows. In Sec. \ref{Sec2}, we present some notations related to group theory, and recall the generalized Simon's problem. 
In Sec. \ref{Sec3}, we characterize the structure of  the generalized Simon's problem in  distributed scenario.  Then,
in Sec. \ref{Sec4} we describe a  distributed  quantum algorithm for  the generalized Simon's problem and give the corresponding analytical procedure.  Furthermore,
 in Sec. \ref{Sec5} we introduce  quantum amplitude amplification technique, and with this technique,
 we in Sec. \ref{Sec6} design an   exact  distributed  quantum algorithm  for  the generalized Simon's problem and prove its correctness. In addition, in Sec. \ref{Sec7}, we compare 
our algorithm with other algorithms. Finally,  we  conclude with a summary in Sec. \ref{Sec8}.
\section{PRELIMINARIES}\label{Sec2}

In this section, we present some notations related to group theory, and recall the generalized Simon's problem.

\subsection{Notations}\label{Notations}
It is known that the generalized Simon's problem is an instance of  the hidden subgroup problem. Below,  we present some of the notations related to group theory.

For $x, y\in\mathbb{Z}_2^n$ with $x=(x_1,\ldots,x_n)$ and $y=(y_1,\ldots,y_n)$, we define
\begin{align}
x+y\coloneqq((x_1+y_1)\bmod2,\ldots, (x_n+y_n)\bmod2).
\end{align}
\begin{align}
x\cdot y\coloneqq(x_1\cdot y_1+\cdots +x_n\cdot y_n)\bmod2.
\end{align}

For any subset $X\subseteq \mathbb{Z}_2^n$, $\langle X\rangle$ denotes the subgroup generated by $X$, i.e.,
\begin{equation}
\langle X\rangle\coloneqq\left\{\sum_{i=1}^k\alpha_i x_{i}|x_{i}\in X, \alpha_i\in\{0,1\}\right\}.
\end{equation}

The set $X$ is linearly independent if $\langle X\rangle\neq\langle Y\rangle$ for any proper subset $Y$ of $X$. Notice that the cardinality $|\langle X\rangle|$ is $2^{|X|}$ if
$X$ is linearly independent.

Let $G$ denote the group $(\{0,1\},\oplus)$, the basis of $G$ is a maximal linearly independent subset of $G$. The cardinality of the basis of $G$ is called its $rank$, denoted by $rank(G)$. If $H$ is a subgroup of $G$, then we denote $H\leq G$.
For $H\leq G$, let
\begin{equation}
H^{\perp}\coloneqq\{g\in G | g\cdot h = 0, \forall h\in H\}.
\end{equation}

Notice that $(H^{\perp})^{\perp}= H$
and $|\langle H\rangle^{\perp}| = 2^{n-|H|}$ if $H$ is linearly independent.

\subsection{The generalized Simon's problem}

The generalized Simon's problem is a special kind of the hidden subgroup problem \cite{kaye_introduction_2007}, which can be described as follows. Consider a function $f:\{0,1\}^n \rightarrow \{0,1\}^m$, where we  promise that  for any $x, y\in {\{0, 1\}}^n$, there is a hidden subgroup $S\leq\mathbb{Z}_2^n$, such that $f(x) = f(y)$ if and only if $x \oplus y \in S$, where $|S|=2^k$ for some $0\leq k\leq n$ $(m\geq n-k)$. For the specific case where $k=1$,  the generalized Simon's problem precisely aligns with Simon's problem.

Denote the basis of $S$ as $\{s_i| s_i\in\{0,1\}^n, 1\leq i\leq k \}$, then we have
\begin{equation}
S=\left\{\sum_{i=1}^k\alpha_i s_{i}\Bigg|s_{i}\in\{0,1\}^{n}, \alpha_i\in\{0,1\}\right\}.
\end{equation}

Suppose we have an oracle that can query the value of function $f$. For any $x \in \{0,1\}^n$ and any $b \in \{0,1\}^m$, if we input $|x\rangle|b\rangle$ into the oracle, then $|x\rangle|b \oplus f(x)\rangle$ is obtained.The goal of the generalized Simon's problem is to find   the hidden subgroup $S$ by performing the minimum number of queries to function $f$. 

The quantum query complexity of  the generalized Simon's problem is $\Theta(n-k)$. The lower bound on the classical (deterministic or randomized) query complexity of  the generalized Simon's problem is $\Omega\left(\max\left\{k,\sqrt{ 2^{n-k}}\right\}\right)$, and its upper bound on the classical query complexity is $O\left(\max\left\{k,\sqrt{k\cdot 2^{n-k}}\right\}\right)$ \cite{WuGSP, KunGSP}.

\section{the generalized Simon's problem in the distributed scenario }\label{Sec3}

In the following, we describe the generalized Simon's problem in distributed scenario and characterize its structure.

The function $f$ corresponding to the generalized Simon's problem  is divided into $2^t$ subfunctions $f_w:\{0,1\}^{n-t}\rightarrow\{0,1\}^m$ as follows.
Let
\begin{equation}\label{General method of function decomposition}
f_w(u)=f(uw),
\end{equation}
where $ u \in \{0,1\}^{n-t}$, $w\in\{0,1\}^t$.

Suppose there are $2^t$ people, each of whom has an oracle $O_{f_w}$ that can query all $f_w(u)=f(uw)$ for any $u \in \{0,1\}^{n-t}$, $w \in \{0,1\}^t$, where $O_{f_w}$ are defined as
\begin{equation}\label{Ofw}
O_{f_w}\ket{u}\ket{b}=\ket{u}\ket{b\oplus f_w(u)},
\end{equation}
where  $u\in\{0,1\}^{n-t}$, $w\in \{0,1\}^t$ and $b\in\{0,1\}^m$.
 
Each person can access $2^{n-t}$ values of $f$.
They need to find the hidden subgroup $S$ by querying their own oracle and communicating with each other as few times as possible.

Below, we further introduce some notations related to the function $f$ corresponding to the generalized Simon's problem. We anticipate that the reader is already familiar with the concept of multisets.

\begin{defi}
	For any $u \in \{0,1\}^{n-t}$, let $G(u)$ denote the multiset  $\{f_w(u)|w \in \{0,1\}^t\}$.
\end{defi}

\begin{defi}
For any $u \in \{0,1\}^{n-t}$, let $N(u,z)=\{w\in\{0,1\}^t|z\in G(u), f_w(u)=z\}$.
\end{defi}

\begin{defi}
  For any $u \in \{0,1\}^{n-t}$, let $S(u)$ represent a string of length $2^tm$ by concatenating all strings $f_w(u)$ $(w\in\{0,1\}^t)$  according to lexicographical order, that is,
  \begin{align}
  S(u)=f_{w_0}(u)f_{w_1}(u)\cdots f_{w_{2^t-1}}(u),
  \end{align}
  where $f_{w_0}(u)\le f_{w_1}(u)\le \ldots \le f_{w_{2^t-1}}(u)\in \{0,1\}^m$, with $w_i\in\{0,1\}^t$ $(0\leq i\leq 2^t-1)$, where $w_i\neq w_j$ for any $i\neq j$ and $\le$ denotes the lexicographical order.
\end{defi}

Let $S$ be the hidden subgroup to be found, and
 denote the basis  of $S$ as 
$\{s_{il}s_{ir}|s_i=s_{il}s_{ir},s_{i}\in\{0,1\}^{n}, s_{il}\in\{0,1\}^{n-t},s_{ir}\in\{0,1\}^{t}, 1\leq i\leq k\}$.

Let
\begin{equation}
S_l=\left\{\sum_{i=1}^k\alpha_i s_{il}\Bigg|s_{il}\in\{0,1\}^{n-t}, \alpha_i\in\{0,1\}\right\},
\end{equation}
then $S_l\leq\mathbb{Z}_2^{n-t}$.

Denote $k_l=rank(S_l)$,   $0\leq k_l\leq \min(k,n-t)$, and the basis  of $S_l$ as $\left\{e_{i}| e_{i}\in\{0,1\}^{n-t}, 1\leq i\leq k_l \right\}$.

Let
\begin{equation}
S_r=\left\{\sum_{i=1}^k\alpha_i s_{ir}\Bigg|s_{ir}\in\{0,1\}^{t}, \alpha_i\in\{0,1\}\right\},
\end{equation}
then $S_r\leq\mathbb{Z}_2^{t}$.

The following theorem concerning $S(u)$ is useful and important.

\begin{theorem}\label{The1} Suppose function $f:\{0,1\}^n \rightarrow \{0,1\}^m$, satisfies that there is a subgroup $S\leq\mathbb{Z}_2^n$  such that $f(x) = f(y)$ if and only if $x \oplus y\in S$. Then
  $\forall u,v \in \{0,1\}^{n-t},S(u)=S(v)$ if and only if $u \oplus v \in S_l$.
\end{theorem}
\begin{proof}
  Based on the properties of multiset, $\forall u,v\in\{0,1\}^{n-t},S(u)=S(v)$ if and only if $G(u) = G(v)$. So our goal is to prove $\forall u,v \in \{0,1\}^{n-t},G(u)=G(v)$ if and only if $u \oplus v\in S_l$.

  (1) $\Longleftarrow$. First, we  prove if $u \oplus v\in S_l$, then $G(u) \subseteq G(v)$. We have $\forall z \in G(u), \exists w \in N(u,z)$ such that $z = f(uw)$. According to the definition of the generalized Simon's problem, $\forall s\in S$, $f(uw \oplus s)=f(uw)=z$. In addition, we have $\forall s\in S$, $\exists s_l\in S_l$, $s_r\in S_r$ such that $s=s_ls_r$.

Thus, we have $\forall s\in S$, $z=f(uw\oplus s)=f((u \oplus s_{l})(w \oplus s_r))$. By the definition of group $S$ and $S_l$, we have $\forall s_l\in S_l$ and $s_ls_r\in S$,  $z=f((u \oplus s_{l})(w\oplus s_r))$. Since $u \oplus v\in S_l$, there $\exists s_l\in S_l$ such that $u \oplus s_l=v$.
Hence, we have $z=f(v(w\oplus s_r)) \in G(v)$. Therefore, $G(u) \subseteq G(v)$ and $|N(u,z)|\leq |N(v,z)|$. Similarly, we can prove that $G(v) \subseteq G(u)$ and $|N(v,z)|\leq |N(u,z)|$. 
As a result, we have $G(u) = G(v)$.

  (2) $\Longrightarrow$. Since $G(u) = G(v)$, we have $\forall z \in G(u), z \in G(v) $. Then we have $\exists w, w' \in \{0,1\}^t$ such that $z = f(uw)$ and $z = f(vw')$.  So we have $f(uw) = f(vw')$. According to the definition of the generalized Simon's problem, we have $uw \oplus vw' =(u\oplus v)(w\oplus w')\in S$.  
Further, by the definition of group $S$ and $S_l$, we have $u \oplus v \in S_l$.
\end{proof}

\section{Distributed quantum algorithm for the generalized Simon's problem}\label{Sec4}

In the following, we begin with  giving related notation, function and operators that are used  in distributed quantum algorithm for the generalized Simon's problem, i.e.,  Algorithm \ref{algorithm1}.

Let $[N]$ represent the set of integers $\{0,1,\cdots, 2^t-1\}$, and let ${\rm BI}:\{0,1\}^t \rightarrow [N]$ be the function to convert a binary string of $t$ bits to an equal decimal integer.

The query operators $O'_{f_w}$ in Algorithm \ref{algorithm1}
are defined as 
\begin{equation}
O'_{f_w}\ket{u}\ket{b}\ket{c}=\ket{u}\ket{b}\ket{c\oplus f_w(u)},
\end{equation}
where $u\in\{0,1\}^{n-t}$, $w\in \{0,1\}^t$, $b\in\{0,1\}^{{\rm BI}(w)}$ and $c\in \{0,1\}^m$.

The  operator $U_{Sort}:\{0,1\}^{2^{t+1}m}\rightarrow\{0,1\}^{2^{t+1}m}$ in Algorithm \ref{algorithm1} 
is defined as 
\begin{equation}\label{U_{Sort}}
\begin{split}
&U_{Sort}\left(\bigotimes_{w\in\{0,1\}^{t}}\ket{f_w(u)}\right)|b\rangle\\
=&\left(\bigotimes_{w\in\{0,1\}^{t}}\ket{f_w(u)}\right)\Ket{b\oplus  S(u)},
\end{split}
\end{equation}
where   $\bigotimes\limits_{w\in\{0,1\}^{t}}\ket{f_w(u)}\triangleq \ket{f_{0^t}(u)}\ket{f_{0^{t-1}1}(u)}\cdots \ket{f_{1^t}(u)}$ and $b\in \{0,1\}^{2^tm}$.

Intuitively, the effect of  $U_{Sort}$ in Algorithm \ref{algorithm1} is to sort the values in the $2^t$ control registers by lexicographical order and XOR to the target register.

In addition, for any operator $A_w$ with $w\in \{0,1\}^t$, we let 
$\prod_{w\in\{0,1\}^t}A_w\triangleq A_{1^t}A_{1^{t-1}0}\cdots A_{0^t}$, $\prod'_{w\in\{0,1\}^t}A_w\triangleq A_{0^t}A_{0^{t-1}1}$ $\cdots A_{1^t}$.

\begin{algorithm}[H]
      \caption{Distributed quantum algorithm for finding the elements  in $S_l^{\perp}$}\label{algorithm1}
      \begin{algorithmic}[1]
\Procedure{DSL}{integer $n$, integer $t$, integer $m$, operator $O_{f_w}$}\strut

 \State $Y\gets \{0^{n-t}\}$;

        \State $|\psi_0\rangle = \Ket{0^{n-t}}\Ket{0^{2^{t+1}m}}$;

        \State $|\psi_1\rangle = \left(H^{\otimes n-t}\otimes I^{\otimes 2^{t+1}m}\right)|\psi_0\rangle$;

\State $\ket{\psi_2}=\prod_{w\in\{0,1\}^t}\left(O'_{f_{w}}\otimes I^{\otimes 2^{t+1}m-{\rm BI}(w)\cdot m-m}\right)\ket{\psi_1}$;
              
       \State  $\ket{\psi_3}=\left(I^{\otimes {n-t}}\otimes U_{sort}\right)\ket{\psi_2}$;

        \State $|\psi_4\rangle=\prod'_{w\in\{0,1\}^t}\left(O'_{f_{w}}\otimes I^{\otimes 2^{t+1}m-{\rm BI}(w)\cdot m-m}\right)\ket{\psi_3}$;

        \State $|\psi_5\rangle=\left(H^{\otimes n-t}\otimes I^{\otimes 2^{t+1}m}\right)|\psi_4\rangle$;

	 \State Measure the first  register, and get an element  $z$;
	    \If{$z\notin \langle Y \rangle$} 
           \State $Y\gets Y\cup \{z\}$. 
          \EndIf
	
  \EndProcedure \strut
      \end{algorithmic}
    \end{algorithm}

In the following, we prove the correctness of   Algorithm \ref{algorithm1}. 
The state after the third step of  Algorithm \ref{algorithm1} is
\begin{equation}
\begin{split}
  |\psi_1\rangle&= \left(H^{\otimes n-t}\otimes I^{\otimes 2^{t+1}m}\right)|\psi_0\rangle\\
  &=\frac{1}{\sqrt{2^{n-t}}}\sum_{u\in\{0,1\}^{n-t}}|u\rangle\Ket{0^{2^{t+1}m}}.
\end{split}
\end{equation}

Then Algorithm \ref{algorithm1} queries each of the oracles to get the following state.
\begin{equation}
\begin{split}
\ket{\psi_2}=&\prod_{w\in\{0,1\}^t}\left(O'_{f_{w}}\otimes I^{\otimes 2^{t+1}m-{\rm BI}(w)\cdot m-m}\right)\ket{\psi_1}\\
=&\frac{1}{\sqrt{2^{n-t}}}\sum_{u\in\{0,1\}^{n-t}}|u\rangle\left(\bigotimes_{w\in\{0,1\}^{t}}\ket{f_w(u)}\right)\Ket{0^{2^tm}}.
\end{split}
\end{equation}

After sorting by using $U_{Sort}$, we have the following state.
\begin{equation}
\begin{split}
\ket{\psi_3}=&\left(I^{\otimes {n-t}}\otimes U_{sort}\right)\ket{\psi_2}\\
=&\frac{1}{\sqrt{2^{n-t}}}\sum_{u\in\{0,1\}^{n-t}}|u\rangle\left(\bigotimes_{w\in\{0,1\}^{t}}\ket{f_w(u)}\right)|S(u)\rangle.
\end{split}
\end{equation}

After that, we query each oracle again and  obtain the following state.
\begin{equation}
\begin{split}
|\psi_4\rangle=&\prod\nolimits_{w\in\{0,1\}^t}'\left(O'_{f_{w}}\otimes I^{\otimes 2^{t+1}m-{\rm BI}(w)\cdot m-m}\right)\ket{\psi_3}\\
=&\frac{1}{\sqrt{2^{n-t}}}\sum_{u\in\{0,1\}^{n-t}}|u\rangle\Ket{0^{2^tm}}|S(u)\rangle.
\end{split}
\end{equation}

After using Hadamard transform on the first  register, in the light of  Theorem \ref{The1}, we   get the following state.
\begin{equation}
\begin{split}
	|\psi_5\rangle=&\left(H^{\otimes n-t}\otimes I^{\otimes 2^{t+1}m}\right)|\psi_4\rangle\\
    =&\frac{1}{2^{n-t}}\sum_{u,z\in\{0,1\}^{n-t}}(-1)^{u\cdot z}
\Ket{z,0^{2^tm},S(u)}\\
    =&\frac{1}{2^{n-t}|S_l|}\sum_{u,z\in\{0,1\}^{n-t}}\sum_{s_l\in S_l}(-1)^{(u\oplus s_l)\cdot z}\\
&
\Ket{z,0^{2^tm},S(u\oplus s_l)}\\
   =&\frac{1}{2^{n-t+k_l}}\sum_{u,z\in\{0,1\}^{n-t}}\sum_{s_l\in S_l}(-1)^{u\cdot z}(-1)^{s_l\cdot z}\\
&
\Ket{z,0^{2^tm},S(u)}\\
   =&\frac{1}{2^{n-t+k_l}}\sum_{u,z\in\{0,1\}^{n-t}}\left(\sum_{s_l\in S_l}(-1)^{s_l\cdot z}\right)(-1)^{u\cdot z}\\
&
\Ket{z,0^{2^tm},S(u)}.
\end{split}
\end{equation}

Note that if there exists $s'_l\in S_l$ such that $s'_l \cdot z=1$, then we have 
\begin{equation}
\begin{split}
\sum_{s_l\in S_l}\left(-1\right)^{s_l\cdot z}
	&=\frac{1}{2}\sum_{s_l\in S_l}\left(\left(-1\right)^{s_l\cdot z}+\left(-1\right)^{\left(s_l\oplus s'_l\right)\cdot z}\right)\\
	&=\frac{1}{2}\sum_{s_l\in S_l}\left(\left(-1\right)^{s_l\cdot z}+\left(-1\right)^{s_l\cdot z}\left(-1\right)^{s'_l\cdot z}\right)\\
	&=\frac{1}{2}\sum_{s_l\in S_l}\left(-1\right)^{s_l\cdot z}\left(1+\left(-1\right)^{s'_l\cdot z}\right)\\
&=0.
\end{split}
\end{equation}

If $z\in S_l^{\perp}$, then $\sum_{s_l\in S_l}(-1)^{s_l\cdot z}=2^{k_l}$, so we have
\begin{equation}
\begin{split}
	|\psi_5\rangle=&\frac{1}{2^{n-t}}\sum_{u\in\{0,1\}^{n-t}}\sum_{z\in S_l^{\perp}}(-1)^{u \cdot z}\Ket{z,0^{2^tm},S(u)}\\
	=&\frac{1}{2^{n-t}}\sum_{z\in S_l^{\perp}}\Ket{z}\sum_{u\in\{0,1\}^{n-t}}(-1)^{u \cdot z}\Ket{0^{2^tm},S(u)}.
\end{split}
\end{equation}

Thus, in line 9 of Algorithm \ref{algorithm1}, after measuring the first register of the state $\ket{\psi_5}$, we can obtain an element $z\in{S_l}^{\perp}$.

In line 9 of Algorithm \ref{algorithm1}, since the result we measure may not be linearly independent of the results we measured earlier, there is no guarantee that $\langle Y\rangle^{\perp} = S_l$ can be obtained for  $Y$ obtained after running Algorithm \ref{algorithm1} iteratively many times.

After running Algorithm \ref{algorithm1} iteratively many times, we denote $S'_l=\langle Y\rangle^{\perp}$.  Denote $k'_l=rank\left(S'_l\right)$,   $k_l\leq k'_l\leq n-t$, and the basis  of $S'_l$ as $\left\{e'_{i}| e'_{i}\in\{0,1\}^{n-t}, 1\leq i\leq k'_l \right\}$. Note that $S'_l$ may not be equal to $S_l$. 

Denote $E_l=S_l\bigcap\{e'_{i}|e'_{i}\in\{0,1\}^{n-t}, 1\leq i\leq k'_l\}$, $\widehat{k_l}=|E_l|$, $0\leq \widehat{k_l}\leq \min\left\{k'_l, 2^{k_l}\right\}$.

\begin{algorithm}[H]
			\caption{Distributed quantum algorithm for finding $S$ }\label{algorithm2}
			\begin{algorithmic}[1]
\Procedure{DS}{integer $n$, integer $t$, subgroup $S'_l$}\strut

				\State Query each oracle $O_{f_{w}}$ once in parallel to get $f\left(0^{n-t}w\right)$ $\left(w \in \{0,1\}^t\right)$;

                   \State Query  oracle $O_{f_{0^t}}$ once  in parallel to get $f\left(e'_{i}0^t\right)$, where $\left\{e'_{i}| e'_{i}\in\{0,1\}^{n-t}, 1\leq i\leq k'_l \right\}$ is  the basis of $S'_l$;
				         
                  \State Find  $v_{i}\in\{0,1\}^t$  in parallel such that $f\left(0^{n-t}v_{i}\right)=f\left(e'_{n_{i}}0^t\right)$, where $e'_{n_{i}}\in E_l$, $1\leq i\leq \widehat{k_l}$, $1\leq n_i\leq k'_l$;

\State Find  $V_{j}=\left\{v\Big|f\left(0^{n-t}\left(\sum_{i=1}^{\widehat{k_l}}\beta_iv_i\right)\right)=f\left(0^{n-t}v\right)\right\}$  in parallel, where $\beta_i\in\{0,1\}$, $j=\sum_{i=1}^{\widehat{k_l}}2^{i-1}\beta_i$;

                  \State $S\gets \bigcup_{j=0 }^{2^{\widehat{k_l}}-1} \left\{\left(\sum_{i=1}^{\widehat{k_l}}\beta_ie'_{n_{i}}\right)v\Big|v\in V_{j} \right\}$;

\State \Return $S$.

  \EndProcedure \strut
			\end{algorithmic}
\end{algorithm}

Denote $S'=\bigcup_{j=0 }^{2^{\widehat{k_l}}-1} \left\{\left(\sum_{i=1}^{\widehat{k_l}}\beta_ie'_{n_{i}}\right)v\Big|v\in V_{j} \right\}$. It will be proved below that  $S$ may not be equal to $S'$, and hence Algorithm \ref{algorithm2} is not exact.

In fact, 
if $S'_l\neq S_l$, then there may  exist  $e_i$ in the basis of $S_l$ and $e_i\notin\bigcup_{j=0 }^{2^{\widehat{k_l}}-1} \left\{\sum_{i=1}^{\widehat{k_l}}\beta_ie'_{n_{i}}\Big|e'_{n_i}\in E_l \right\}$, where  $E_l=S_l\bigcap\left\{e'_{i}|e'_{i}\in\{0,1\}^{n-t}, 1\leq i\leq k'_l\right\}$, $\widehat{k_l}=|E_l|$, $\beta_i\in\{0,1\}$ and $j=\sum_{i=1}^{\widehat{k_l}}2^{i-1}\beta_i$.
Let $s=e_iv_i\in S$, where $e_i\in S_l\setminus\bigcup_{j=0 }^{2^{\widehat{k_l}}-1} \left\{\sum_{i=1}^{\widehat{k_l}}\beta_ie'_{n_{i}}\Big|e'_{n_i}\in E_l \right\}$ and $v_i\in S_r$. 
Thus, $s=e_iv_i\notin S'$, which indicates that there is a case where $S$ is not equal to $S'$, i.e., Algorithm \ref{algorithm2} is not exact.

\section{Quantum amplitude amplification}\label{Sec5}

By means of the work of Cai and Qiu \cite{cai_optimal_2018}, we also require a similar method, i.e., the use of quantum amplitude amplification technique to make the distributed quantum algorithm for solving the generalized Simon's problem exactly. 

To make Algorithm  \ref{algorithm1} exact, we add
a post-processing subroutine after line 8 of Algorithm \ref{algorithm1} to ensure $(Y\setminus \{0^{n-t}\})\cup\{z\}$ is always linearly
independent when we get the measured result $z$ of the first  register.

Denote by
\begin{equation}
T=\left\{x_i\Big|\bigcup\nolimits_{i=1}^{2^{n-t-k_l}}\bigcup\nolimits_{s_l\in S_l}\{x_i\oplus s_l\}=\{0,1\}^{n-t}\right\},
\end{equation}
\begin{equation}
S(T)=\{S(u)|u\in T\},
\end{equation}
and let
\begin{equation}
\Ket{{S_l}^{\perp}, 0^{2^tm}, S(T)}
\end{equation}
denote the  state after line 8 of Algorithm  \ref{algorithm1}. 

Let $\mathcal{A}: \{0,1\}^{n-t+2^{t+1}m} \rightarrow \{0,1\}^{n-t+2^{t+1}m}$ denote the combined unitary operators from line 4 to line 8 in Algorithm  \ref{algorithm1}, i.e., $\mathcal{A}$ is defined as
\begin{equation}\label{A}
\begin{split}
\mathcal{A}=&\left(H^{\otimes n-t}\otimes I^{\otimes 2^{t+1}m}\right)\\
&\left(\prod\nolimits_{w\in\{0,1\}^t}'\left(O'_{f_{w}}\otimes I^{\otimes 2^{t+1}m-{\rm BI}(w)\cdot m-m}\right)\right)\\
&\left(I^{\otimes {n-t}}\otimes U_{sort}\right)\\
&\left(\prod\nolimits_{w\in\{0,1\}^t}\left(O'_{f_{w}}\otimes I^{\otimes 2^{t+1}m-{\rm BI}(w)\cdot m-m}\right)\right)\\
&\left(H^{\otimes n-t}\otimes I^{\otimes 2^{t+1}m}\right).
\end{split}
\end{equation}

Define $\mathcal{R}_{0}(\phi):\{0,1\}^{n-t+2^{t+1}m} \rightarrow \{0,1\}^{n-t+2^{t+1}m}$ as
\begin{equation}
\begin{split}
		\label{eq:R_0}
		&\mathcal{R}_{0}(\phi)\Ket{x,b}\\=&
		\begin{cases}
		\Ket{x,b}\text{, } & x\neq 0^{n-t} \text{ or } b\neq 0^{2^{t+1}m}\text{;}\\
		e^{i\phi}\Ket{x,b}\text{, } & x=0^{n-t} \text{ and } b=0^{2^{t+1}m}\text{.}
          \end{cases}
\end{split}
\end{equation}

Define $\mathcal{R}_{\mathcal{A}}(\varphi, Y):\{0,1\}^{n-t} \rightarrow \{0,1\}^{n-t}$ as
\begin{align}
		\label{eq:R_A}
		&\mathcal{R}_{\mathcal{A}}(\varphi, Y)\Ket{x}=
		\begin{cases}
		e^{i\varphi}\Ket{x}\text{, } & x\notin \langle Y \rangle\text{;}\\
		\Ket{x}\text{, } & x\in \langle Y \rangle\text{.}
		\end{cases}
\end{align}

Then by  $\mathcal{R}_{0}(\phi)$  and $\mathcal{R}_{\mathcal{A}}(\varphi, Y)$, we define the quantum amplitude amplification operator $\mathcal{Q}: \{0,1\}^{n-t+2^{t+1}m} \rightarrow \{0,1\}^{n-t+2^{t+1}m}$ as 
\begin{equation}\label{eq:Q}
	\mathcal{Q}
	=-\mathcal{A}\mathcal{R}_{0}(\phi)
	\mathcal{A}^{\dagger}\left(\mathcal{R}_{\mathcal{A}}(\varphi, Y)\otimes I^{\otimes 2^{t+1}m}\right).
\end{equation}

Denote by
 \begin{equation}
  X={S_l}^{\perp}\setminus \langle Y \rangle.
  \end{equation}

\begin{defi}
Let $\Ket{\Psi_{X}}$ denote the projection onto the good state subspace, that is, the subspace spanned by $\big\{\Ket{x,b} \mid x\in X, b\in \{0,1\}^{2^{t+1}m}\big\}$.
\end{defi}

\begin{defi}
Let $\Ket{\Psi_{Y}}$ denote the projection onto the bad state subspace, that is, the subspace spanned by $\big\{\Ket{y,b} \mid y\in \langle Y \rangle, b\in \{0,1\}^{2^{t+1}m} \big\}$.
\end{defi}

We have
\begin{equation}
	\Ket{{S_l}^{\perp},0^{2^tm},S(T)}=\Ket{\Psi_{X}}+\Ket{\Psi_{Y}}.
\end{equation}

 In order to make  Algorithm  \ref{algorithm1} exact, the crucial step is to eliminate all states in $\langle Y \rangle$ from the first register. In quantum
amplitude amplification process, one can achieve this by choosing appropriate $\phi,\varphi \in \mathbb{R}$ such that after applying $\mathcal{Q}$ on $\Ket{{S_l}^{\perp},0^{2^tm},S(T)}$, the
amplitudes of all states in $\langle Y \rangle$ of the first register become zero.

In the following, we present  a proposition related to the  operator $\mathcal{Q}$  acting on the state $\Ket{{S_l}^{\perp},0^{2^tm},S(T)}$, which is proved in Appendix \ref{proof of Theorem 2}.

 \setcounter{theorem}{0}
\begin{Pp}\label{Pp1} 
Let  $\phi=2\arctan{\left(\sqrt{\frac{2^{n-t-r-k_l}}{3\cdot 2^{n-t-r-k_l}-4}}\right)}$, $\varphi = \arccos{\left(\frac{2^{n-t-r-k_l-1}-1}{2^{n-t-r-k_l}-1}\right)}$, where $r=|Y|-1$. Then
	\begin{align}
	     \mathcal{Q}\Ket{S_l^{\perp},0^{2^tm},S(T)} =\Ket{\Psi_{X}}.
	\end{align}
\end{Pp}

We describe the quantum amplitude amplification algorithm used to measure good states as follows. In fact, since we do not know the rank of  group $S$, i.e., $k_l$, we assume that the rank of  group $S$ is $d_l$, where $0\leq d_l\leq \min(k,n-t)$.

\begin{algorithm}[H]
	\caption{\strut Quantum amplitude amplification for measuring good states}\label{algorithm3}
	\begin{algorithmic}[1]
	\Procedure{QAA}{registers $\Ket{S_l^{\perp},0^{2^tm},S(T)}$, integer $n$, integer $m$, integer $t$, integer $d_l$, operator $\mathcal{A}$, set $Y$ }
	\State $r\gets |Y|-1$;
	\State $\phi\gets 2\arctan{\left(\sqrt{\frac{2^{n-t-r-d_l}}{3\cdot 2^{n-t-r-d_l}-4}}\right)}$;
	\State $\varphi\gets \arccos{\left(\frac{2^{n-t-r-d_l-1}-1}{2^{n-t-r-d_l}-1}\right)}$;
	\State Apply $\mathcal{Q}$ to $\Ket{S_l^{\perp},0^{2^tm}, S(T)}$, where
$\mathcal{Q}=-\mathcal{A}\mathcal{R}_{0}(\phi)
	\mathcal{A}^{\dagger}\left(\mathcal{R}_{\mathcal{A}}(\varphi, Y)\otimes I^{\otimes 2^{t+1}m}\right)$;
	\State Measure the first register, get the result $z$;
	\State \Return $z$.
	\EndProcedure \strut
	\end{algorithmic}
\end{algorithm}

\section{Exact distributed  quantum algorithm for the generalized Simon's problem}\label{Sec6}

In this section, we  first design  an  exact distributed  quantum algorithm for finding $S_l$, i.e., Algorithm \ref{algorithm4}, which combines  Algorithm \ref{algorithm1} and Algorithm \ref{algorithm3}. After finding $S_l$, we design Algorithm \ref{algorithm5} to find the hidden subgroup $S$ exactly.

We describe the main design idea for Algorithm \ref{algorithm4} as follows. First, we  use Algorithm \ref{algorithm1} to ensure that the state of the first register is in $S_l^{\perp}$, and then we utilize quantum amplitude amplification technique \cite{BHMT02} to ensure that the measured result of the first register is not in $\langle Y\rangle$.

Since $rank(S_l)$ is undetermined, we assume it is $d_l$ and initialise $d_l=0$. 
For a given $d_l$, if $d_l\neq k_l$, then during the iterative run of  Algorithm \ref{algorithm4},  it must be obtained that $z\in\langle Y\rangle$. In this case, the value of $d_l$ is increased by 1.  If $z\notin\langle Y\rangle$, then add $z$ to $Y$. After that,  Algorithm \ref{algorithm4} is run repeatedly.

 If $|Y|=n-t+1-d_l$, this means that $d_l$ has been increased to equal $k_l$. We have also obtained the set $Y$, which satisfies $\langle Y \rangle=S_l^{\perp}$. Eventually, by solving the system of exclusive-or equations, we can obtain $S_l= \langle Y \rangle ^{\perp}$.

\begin{figure*}[htbp]
  \centering
  \includegraphics[width=\textwidth]{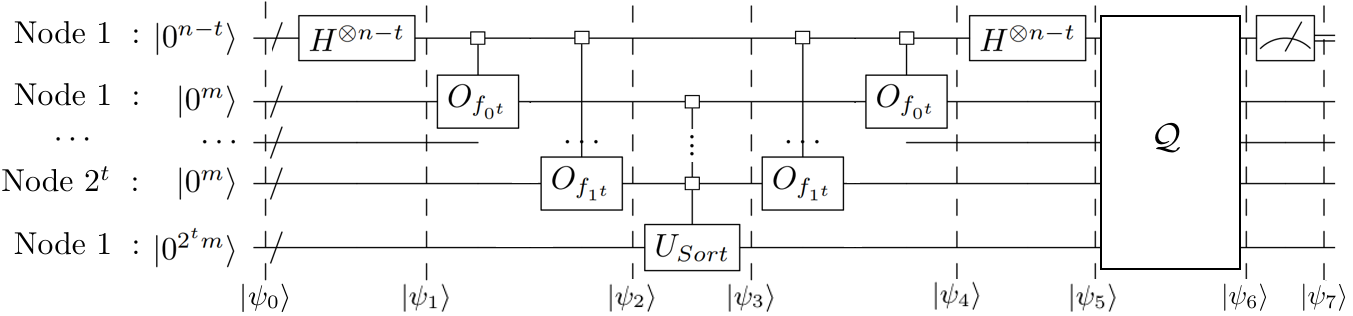}
  \caption{The circuit for the quantum part of   exact  distributed   quantum algorithm for finding $S_l$  (Algorithm \ref{algorithm4}).}
  \label{algorithm4_circuit}
\end{figure*}

\begin{algorithm}[H]
	\caption{\strut Exact  distributed quantum algorithm for finding $S_l$ }
	\label{algorithm4}
	\begin{algorithmic}[1]
	\Procedure{EDSL}{integer $n$, integer $t$, integer $m$, operator $O_{f_w}$}\strut
     \State $d_l\gets 0$;
	\State $Y\gets \{0^{n-t}\}$;
	  \Repeat
	  \State Prepare registers $\Ket{0^{n-t+2^{t+1}m}}$;

      \State  Apply  operator $\mathcal{A}$ on $\Ket{0^{n-t+2^{t+1}m}}$, where $\mathcal{A}$ is defined in Eq. (\ref{A});

	  \State $z\gets$ \Call{QAA}{$\mathcal{A}\Ket{0^{n-t+2^{t+1}m}}$, $n$, $m$, $t$, $d_l$, $\mathcal{A}$, $Y$};

\If{ $z\in\langle Y \rangle$} 
                          \State $d_l\gets d_l+1$;       
                          \Else 
                          \State $Y\gets Y\cup \{z\}$; 
                      \EndIf

	  \Until{$|Y|=n-t+1-d_l$};
	\State Solve the system of exclusive-or equations, get $S_l= \langle Y \rangle ^{\perp}$;
	\State \Return $S_l$.
	\EndProcedure \strut
	\end{algorithmic}
\end{algorithm}

In the following, we present Algorithm \ref{algorithm5}. The main design idea of Algorithm \ref{algorithm5} is to first find the associated string corresponding to the base of  group $S_l$, and then concatenate the strings formed by the base of  group $S_l$ with their corresponding associated strings, and finally merge them to form the hidden subgroup $S$ exactly.

\begin{algorithm}[H]
			\caption{Exact distributed  quantum algorithm for finding $S$ }\label{algorithm5}
			\begin{algorithmic}[1]
\Procedure{EDS}{integer $n$, integer $t$, subgroup $S_l$}\strut

				\State Query each oracle $O_{f_{w}}$ once in parallel to get $f\left(0^{n-t}w\right)$ $\left(w \in \{0,1\}^t\right)$;

                   \State Query  oracle $O_{f_{0^t}}$ once  in parallel to get $f\left(e_{i}0^t\right)$, where $\left\{e_{i}| e_{i}\in\{0,1\}^{n-t}, 1\leq i\leq k_l \right\}$ is  the basis of $S_l$;
				         
                  \State Find  $v_i\in\{0,1\}^t$  in parallel such that $f\left(0^{n-t}v_i\right)=f\left(e_i0^t\right)$;

\State Find  $V^*_j=\left\{v\Big|f\left(0^{n-t}\left(\sum_{i=1}^{k_l}\gamma_iv_i\right)\right)=f\left(0^{n-t}v\right)\right\}$  in parallel, where $\gamma_i\in\{0,1\}$, $j=\sum_{i=1}^{k_l}2^{i-1}\gamma_i$;

                  \State $S\gets \bigcup_{j=0 }^{2^{k_l}-1} \left\{\left(\sum_{i=1}^{k_l}\gamma_ie_i\right)v\Big|v\in V^*_j \right\}$;

\State \Return $S$.

  \EndProcedure \strut
			\end{algorithmic}
\end{algorithm}

In the following, we prove the correctness of Algorithm \ref{algorithm4} and Algorithm \ref{algorithm5}.

First, we prove the correctness of Algorithm \ref{algorithm4}.
In line 6 of Algorithm \ref{algorithm4}, since operator $\mathcal{A}$ is defined in Eq. (\ref{A}), which is the combined unitary operators from line 4 to line 8 in Algorithm  \ref{algorithm1}, the following equation can be obtained  from the proof of correctness of Algorithm \ref{algorithm1}.

\begin{equation}
\begin{split}
&\mathcal{A}\Ket{0^{n-t+2^{t+1}m}}\\
=&\ket{\psi_5}\\
	=&\frac{1}{2^{n-t}}\sum_{z\in S_l^{\perp}}\Ket{z}\sum_{u\in\{0,1\}^{n-t}}(-1)^{u \cdot z}\Ket{0^{2^tm},S(u)}\\
=&\Ket{{S_l}^{\perp},0^{2^tm},S(T)}.
\end{split}
\end{equation}

Applying $\mathcal{Q}$ to $\mathcal{A}\Ket{0^{n-t+2^{t+1}m}}$,  we have the state $\ket{\psi_6}$ in FIG. \ref{algorithm4_circuit} as
\begin{equation}
\begin{split}
	\ket{\psi_6}=&\mathcal{Q}\mathcal{A}\Ket{0^{n-t+2^{t+1}m}}\\
=&\mathcal{Q}\Ket{{S_l}^{\perp},0^{2^tm},S(T)},
\end{split}
\end{equation}
where $\mathcal{Q}
	=-\mathcal{A}\mathcal{R}_{0}(\phi)
	\mathcal{A}^{\dagger}\left(\mathcal{R}_{\mathcal{A}}(\varphi, Y)\otimes I^{\otimes 2^{t+1}m}\right)$, $\phi=2\arctan{\left(\sqrt{\frac{2^{n-t-r-k_l}}{3\cdot 2^{n-t-r-k_l}-4}}\right)}$, $\varphi = \arccos{\left(\frac{2^{n-t-r-k_l-1}-1}{2^{n-t-r-k_l}-1}\right)}$, $r=|Y|-1$.
	
Based on the Proposition \ref{Pp1}, we have
\begin{equation}
	\ket{\psi_6}=\Ket{\Psi_{X}}.
\end{equation}

After measuring on the first  register, we can get an element that is in $S_l^{\perp}\setminus \langle Y \rangle$. After $n-t-k_l$ repetitions of  Algorithm \ref{algorithm4}, we can obtain $n-t-k_l$ elements in $S_l^{\perp}$. Then, using the classical Gaussian elimination method, we can obtain $S_l$.

If we have already found out $S_l$, we can use Algorithm \ref{algorithm5} to find out the hidden subgroup $S$. In the following, we prove the correctness of Algorithm \ref{algorithm5}.

Denote by
\begin{equation}
\begin{split}
S^*= \bigcup_{j=0 }^{2^{k_l}-1} \left\{\left(\sum_{i=1}^{k_l}\gamma_ie_i\right)v\Bigg|v\in V^*_j \right\},
\end{split}
\end{equation}
where $\gamma_i\in\{0,1\}$, $j=\sum_{i=1}^{k_l}2^{i-1}\gamma_i$.

To prove the correctness of Algorithm \ref{algorithm5}, we prove that $S = S^*$.

First, we prove that $S^*\subseteq S$. 

Let $s^*=\left(\sum_{i=1}^{k_l}\gamma_ie_i\right)v$ be any element belonging to $S^*$. Since $f\left(0^{n-t}v_i\right)=f\left(e_i0^t\right)$, by the definition of the generalized Simon's problem, we have 
\begin{equation}
f\left(0^{n-t}\left(\sum_{i=1}^{k_l}\gamma_iv_i\right)\right)=f\left(\left(\sum_{i=1}^{k_l}\gamma_ie_i\right)0^t\right). 
\end{equation}

Moreover,  for any $v\in V_j$, there is 
\begin{equation}
f\left(0^{n-t}v\right)=f\left(0^{n-t}\left(\sum_{i=1}^{k_l}\gamma_iv_i\right)\right). 
\end{equation}

Therefore, we have $f\left(0^{n-t}v\right)=f\left(\left(\sum_{i=1}^{k_l}\gamma_ie_i\right)0^t\right)$. According to the definition of the generalized Simon's problem, we have $s^*=\left(\sum_{i=1}^{k_l}\gamma_ie_i\right)v\in S$. Thus, $S^*\subseteq S$.

Then, we prove that $S\subseteq S^*$. 

Let $s=s_ls_r$ be any element belonging to $S$, where $s_l\in S_l$, $s_r\in S_r$. 

Since $\left\{e_{i}| e_{i}\in\{0,1\}^{n-t}, 1\leq i\leq k_l \right\}$  is  the basis of $S_l$, we also have
\begin{equation}
 s_l\in\left\{\sum_{i=1}^{k_l}\gamma_ie_i\Big|e_i\in\{0,1\}^{n-t}, \gamma_i\in\{0,1\}\right\}. 
\end{equation}

By the definition of the generalized Simon's problem, we have $f(0^{n-t}s_r)=f(s_l0^t)$.

 Furthermore, there exists $\gamma_i\in\{0,1\}$ such that 
\begin{equation}
\begin{split}
f(0^{n-t}s_r)=&f(s_l0^t)\\
=&f\left(\left(\sum_{i=1}^{k_l}\gamma_ie_i\right)0^t\right)\\
=&f\left(0^{n-t}\left(\sum_{i=1}^{k_l}\gamma_iv_i\right)\right). 
\end{split}
\end{equation}

From the definition of $V^*_j$, we have $s_r\in V^*_j$. Thus, $s\in S^*$.  Consequently, $S\subseteq S^*$.

\section{Comparisons with other algorithms}\label{Sec7}

First, we compare Algorithm \ref{algorithm1} with the  distributed  classical randomized algorithm for solving the generalized Simon's problem. Based on the previous analysis, Algorithm \ref{algorithm1} needs $O(n-t-k_l)$ queries for finding group $S_l$. However, in order to find group $S_l$, the  distributed  classical randomized algorithm needs to query oracles $\Omega\left(\max\left\{k_l,\sqrt{2^{n-t-k_l}}\right\}\right)$ times.

\begin{table}[H]   
\begin{center}   

\caption{Comparison of Algorithm \ref{algorithm1} with distributed classical randomized algorithm.}  
\begin{tabular}{|c|c|}   
\hline   \textbf{Algorithms} & \textbf{Query complexity} \\   
\hline   Algorithm \ref{algorithm1} & $O(n-t-k_l)$ \\ 
\hline  \makecell[c]{Distributed classical  \\ randomized algorithm}  & $\Omega\left(\max\left\{k_l,\sqrt{2^{n-t-k_l}}\right\}\right)$ \\  
\hline   
\end{tabular}   
\end{center}   
\end{table}

Therefore, Algorithm \ref{algorithm1} has the advantage of exponential acceleration compared with  the  distributed classical randomized algorithm.

Second, we compare Algorithm \ref{algorithm4} with the  distributed  classical deterministic algorithm for solving the generalized Simon's problem. According to the previous analysis, Algorithm \ref{algorithm4} needs $O(n-t)$ queries  for finding group $S_l$. However, in order to find group $S_l$, the  distributed  classical deterministic algorithm needs to query oracles $\Omega\left(\max\left\{k_l,\sqrt{2^{n-t-k_l}}\right\}\right)$ times.

\begin{table}[H]   
\begin{center}   

\caption{Comparison of Algorithm \ref{algorithm4} with distributed classical deterministic algorithm.}  
\begin{tabular}{|c|c|}   
\hline   \textbf{Algorithms} & \textbf{Query complexity} \\   
\hline   Algorithm \ref{algorithm4} & $O(n-t)$ \\ 
\hline  \makecell[c]{Distributed classical\\ deterministic  algorithm}  & $\Omega\left(\max\left\{k_l,\sqrt{2^{n-t-k_l}}\right\}\right)$ \\  
\hline   
\end{tabular}   
\end{center}   
\end{table}

Hence, Algorithm \ref{algorithm4} has the advantage of exponential acceleration compared with  the  distributed classical deterministic algorithm.

Third, we compare Algorithm \ref{algorithm1} and Algorithm \ref{algorithm4} with the centralized quantum algorithm for solving  the generalized Simon's problem. In Algorithm \ref{algorithm1} and Algorithm \ref{algorithm4}, the number of actual functioning qubits  for  each oracle is only $n-t+m$. However, in the centralized quantum algorithm, the number of actual functioning qubits  for each oracle is $n+m$. 

\begin{table}[H]   
\begin{center}   

\caption{Comparison of Algorithm \ref{algorithm1} and Algorithm \ref{algorithm4} with the centralized quantum algorithm for solving  the generalized Simon's problem.}  
\begin{tabular}{|c|c|}   
\hline    \makecell[c]{\textbf{Algorithms}} & \makecell[c]{\textbf{The number of actual}\\ \textbf{functioning qubits for}\\  \textbf{each oracle } } \\   
\hline   Algorithm \ref{algorithm1} & $n-t+m$ \\ 
\hline   Algorithm \ref{algorithm4} & $n-t+m$ \\ 
\hline  \makecell[c]{The centralized quantum\\ algorithm for solving the\\  generalized Simon's problem} &  $n+m$ \\  
\hline   
\end{tabular}   
\end{center}   
\end{table}

Consequently, Algorithm \ref{algorithm1} and Algorithm \ref{algorithm4} facilitate the reduction of circuit depth and the physical implementation of  algorithm in the NISQ era.

Finally, we compare Algorithm \ref{algorithm2} and Algorithm \ref{algorithm5} with the best distributed quantum algorithm  for  Simon's problem  proposed previously \cite{Tan2022DQCSimon}. Algorithm \ref{algorithm2} and Algorithm \ref{algorithm5} can not only solve Simon's problem, but also  the generalized Simon's problem. In particular, Algorithm \ref{algorithm5} is exact. However, the   algorithm   \cite{Tan2022DQCSimon}  cannot  solve  the generalized Simon's problem and is not exact.

\begin{table}[H]   
\begin{center}   
\caption{Comparison of Algorithm \ref{algorithm2} and Algorithm \ref{algorithm5} with  algorithm in \cite{Tan2022DQCSimon}.}  
\label{table:1} 
\begin{tabular}{|c|c|c|}   
\hline   \textbf{Algorithms} & \textbf{Precision}& \textbf{Generalisability} \\   
\hline   Algorithm \ref{algorithm2} & Inaccuracy&\makecell[c]{The generalized \\ Simon's problem} \\ 
\hline   Algorithm \ref{algorithm5} & Exact& \makecell[c]{The generalized \\ Simon's problem} \\ 
\hline   \makecell[c]{The  algorithm  in \cite{Tan2022DQCSimon}} &Inaccuracy& Simon's problem\\  
\hline   
\end{tabular}   
\end{center}   
\end{table}

As a result, Algorithm \ref{algorithm2} and Algorithm \ref{algorithm5} have better generalisability.  Algorithm \ref{algorithm5} not only has better generalisability, but also has the advantage of being exact.

\section{Conclusion}\label{Sec8}

In this paper, we have characterized the structure of  the generalized Simon's problem in  distributed scenario. Based on the structure, we have designed a corresponding distributed quantum algorithm. Then we have further utilized quantum amplitude amplification technique to make our algorithm exact.
The number of actual functioning qubits  for  each oracle in our algorithm is reduced, which reduces the circuit depth and helps reduce circuit noise, making our algorithm easier to be implemented in the current NISQ era.

Our  algorithm has the advantage of exponential acceleration compared with the  distributed
 classical algorithm.
Compared to the centralized quantum algorithm for  the generalized Simon's problem, the oracle in our algorithm is easier to be implemented, which is an important advantage for implementing quantum query algorithms.
 Compared with the best distributed quantum algorithm  for Simon's problem proposed previously, the  exact distributed quantum algorithm we designed for  the generalized Simon's problem has the advantage of better generalisability and exactness.

We found that characterizing the essential structure of  problem is crucial for designing  corresponding  exact  distributed  quantum algorithm. In future research work, our algorithm may be instructive for designing   exact  distributed quantum algorithms for solving the hidden subgroup problem. Furthermore, the ideas and methods employed in the design of our algorithm may also be useful for the design of  exact  distributed  quantum algorithms for solving other problems.

\section*{Acknowledgements}
This work is supported in part by the National Natural Science Foundation of China (Nos. 61876195, 61572532), and the Natural Science Foundation of Guangdong Province of China (No. 2017B030311011).

\appendix
\section{Proof of Proposition 1}
\label{proof of Theorem 2}
\setcounter{theorem}{0}
\begin{Pp}\label{Pp1} 
Let  $\phi=2\arctan{\left(\sqrt{\frac{2^{n-t-r-k_l}}{3\cdot 2^{n-t-r-k_l}-4}}\right)}$, $\varphi = \arccos{\left(\frac{2^{n-t-r-k_l-1}-1}{2^{n-t-r-k_l}-1}\right)}$, where $r=|Y|-1$. Then
	\begin{align*}
	     \mathcal{Q}\Ket{S_l^{\perp},0^{2^tm},S(T)} =\Ket{\Psi_{X}}.
	\end{align*}
\end{Pp}
\begin{proof}
	From Eq. (\ref{eq:R_0}), we can write $\mathcal{R}_{0}(\phi)$ as follows.
	\begin{equation}
	\begin{split}
		\mathcal{R}_{0}(\phi) =&I^{\otimes n-t + 2^{t+1}m}\\ 
		& - \left(1 - e^{i\phi}\right)\Ket{0^{n-t}, 0^{2^{t+1}m}}\Bra{0^{n-t}, 0^{2^{t+1}m}}.
      \end{split}
	\end{equation}
	
	From the definitions of $\mathcal{R}_{\mathcal{A}}(\varphi, Y)$, $\Ket{\Psi_{X}}$ and $\Ket{\Psi_{Y}}$, we have
	\begin{align}
		\left(\mathcal{R}_{\mathcal{A}}(\varphi, Y)\otimes I^{\otimes {2^{t+1}}m}\right)\Ket{\Psi_{X}} &= e^{i\varphi}\Ket{\Psi_{X}}.
	\end{align}
	\begin{align}
		\left(\mathcal{R}_{\mathcal{A}}(\varphi, Y)\otimes I^{\otimes {2^{t+1}}m}\right)\Ket{\Psi_{Y}} &= \Ket{\Psi_{Y}}.
	\end{align}
	
	Let $\mathcal{U}(\mathcal{A}, \phi) = -\mathcal{A}\mathcal{R}_{0}(\phi)\mathcal{A}^{\dagger}$, then based on Eq. (\ref{eq:Q}), $\mathcal{Q}$ can be written as
	\begin{equation}
		\mathcal{Q} = \mathcal{U}(\mathcal{A}, \phi)\left(\mathcal{R}_{\mathcal{A}}(\varphi, Y)\otimes I^{\otimes {2^{t+1}}m}\right).
	\end{equation}
	
	For $\mathcal{U}(\mathcal{A}, \phi)$, we have
	\begin{equation}
	  \begin{split}
		\mathcal{U}(\mathcal{A}, \phi) =&-\mathcal{A}\mathcal{R}_{0}(\phi)\mathcal{A}^{\dagger}\\
				   =&-\mathcal{A}\left(I^{\otimes n-t + 2^{t+1}m}\right. - \\
				   &\left.\left(1 - e^{i\phi}\right)\Ket{0^{n-t}, 0^{2^{t+1}m}}\Bra{0^{n-t}, 0^{2^{t+1}m}}\right)\mathcal{A}^{\dagger}\\
				   =&  \left(1 - e^{i\phi}\right)\left(\mathcal{A}\Ket{0^{n-t}, 0^{2^{t+1}m}}\Bra{0^{n-t}, 0^{2^{t+1}m}}\mathcal{A}^{\dagger}\right)\\
				    &- I^{\otimes n-t + 2^{t+1}m}\\
				  =& \left(1 - e^{i\phi}\right)\Ket{K^{\perp}, 0^{2^tm}, S(T)}\Bra{K^{\perp}, 0^{2^tm}, S(T)}\\
				  & - I^{\otimes n-t + 2^{t+1}m}\\
				   =& \left(1 - e^{i\phi}\right)\big(\Ket{\Psi_{X}}+\Ket{\Psi_{Y}}\big)\big(\Bra{\Psi_{X}}+\Bra{\Psi_{Y}}\big)\\
				  & - I^{\otimes n-t + 2^{t+1}m}.
        \end{split}
	\end{equation}
	
Since $| \langle Y \rangle | = 2 ^ {r} $and $| S_l ^ {\perp} | = 2 ^ {n-t - k_l} $, according to the definition of $\Ket {\Psi_ {X}} $ and $\Ket {\Psi_ {Y}} $, we have
\begin{equation}
\begin{split}
	\Braket{\Psi_{X}|\Psi_{X}} &= 1-2^{r+k_l+t-n}.\\
     \Braket{\Psi_{Y}|\Psi_{Y}} &= 2^{r+k_l+t-n}.\\
	\Braket{\Psi_{X}|\Psi_{Y}} &= 0.
\end{split}
\end{equation}

Thus, we can obtain the following equations.
\begin{equation}\label{thm2eq1}
\begin{split}
	\mathcal{Q}\Ket{\Psi_{X}} =& \mathcal{U}(\mathcal{A}, \phi)\left(\mathcal{R}_{\mathcal{A}}(\varphi, Y)\otimes I^{\otimes 2^{t+1}m}\right)\Ket{\Psi_{X}}\\
							  =& e^{i\varphi}\mathcal{U}(\mathcal{A}, \phi)\Ket{\Psi_{X}}\\
							  =& e^{i\varphi}\Big(\left(1 - e^{i\phi}\right)\big(\Ket{\Psi_{X}}+\Ket{\Psi_{Y}}\big)\big(\Bra{\Psi_{X}}+\Bra{\Psi_{Y}}\big) \\
							& - I^{\otimes n-t + 2^{t+1}m}\Big)\Ket{\Psi_{X}}\\
							  =& e^{i\varphi}\left(1 - e^{i\phi}\right)\big(\Ket{\Psi_{X}}+\Ket{\Psi_{Y}}\big)\big(\Bra{\Psi_{X}}+\Bra{\Psi_{Y}}\big)\\
							  &\Ket{\Psi_{X}} - e^{i\varphi}\Ket{\Psi_{X}}\\
							  =& e^{i\varphi}\left(1 - e^{i\phi}\right)\Braket{\Psi_{X}|\Psi_{X}}\big(\Ket{\Psi_{X}}+\Ket{\Psi_{Y}}\big)\\
							   &- e^{i\varphi}\Ket{\Psi_{X}}\\
							  =& e^{i\varphi}\left(1 - e^{i\phi}\right)\left(1-2^{r+k_l+t-n}\right)\big(\Ket{\Psi_{X}}+\Ket{\Psi_{Y}}\big)\\
							  & - e^{i\varphi}\Ket{\Psi_{X}}\\
							  =& e^{i\varphi}\left((1-e^{i\phi})(1-2^{r+k_l+t-n})-1\right)\Ket{\Psi_{X}}\\
							  & + e^{i\varphi}(1-e^{i\phi})(1-2^{r+k_l+t-n})\Ket{\Psi_{Y}}.
\end{split}						  
\end{equation}

\begin{equation}\label{thm2eq2}
\begin{split}
	\mathcal{Q}\Ket{\Psi_{Y}} =& \mathcal{U}(\mathcal{A}, \phi)\left(\mathcal{R}_{\mathcal{A}}(\varphi, Y)\otimes
	                             I^{\otimes 2^{t+1}m}\right)\Ket{\Psi_{Y}}\\
							  =& \mathcal{U}(\mathcal{A}, \phi)\Ket{\Psi_{Y}}\\
							  =& \Big(\left(1 - e^{i\phi}\right)\big(\Ket{\Psi_{X}}+\Ket{\Psi_{Y}}\big)\big(\Bra{\Psi_{X}}+\Bra{\Psi_{Y}}\big)\\
							   &- I^{\otimes n-t + 2^{t+1}m}\Big)\Ket{\Psi_{Y}}\\
							  =& \left(1 - e^{i\phi}\right)\big(\Ket{\Psi_{X}}+\Ket{\Psi_{Y}}\big)\big(\Bra{\Psi_{X}}+\Bra{\Psi_{Y}}\big)\Ket{\Psi_{Y}}\\
							  & - \Ket{\Psi_{Y}}\\
							  =& \left(1 - e^{i\phi}\right)\Braket{\Psi_{Y}|\Psi_{Y}}\big(\Ket{\Psi_{X}}+\Ket{\Psi_{Y}}\big) - \Ket{\Psi_{Y}}\\
							  =& \left(1 - e^{i\phi}\right)2^{r+k_l+t-n}\big(\Ket{\Psi_{X}}+\Ket{\Psi_{Y}}\big) - \Ket{\Psi_{Y}}\\
							  =& (1-e^{i\phi})2^{r+k_l+t-n}\Ket{\Psi_{X}} 
							  \\&-\left((1-e^{i\phi})(1-2^{r+k_l+t-n})+e^{i\phi} \right)\Ket{\Psi_{Y}}.
\end{split}						  
\end{equation}

By making sure the resulting superposition $\mathcal{Q}(\Ket{\Psi_{X}}+\Ket{\Psi_{Y}})$ has inner product zero with $\Ket{\Psi_{Y}}$, then based on Eq. (\ref{thm2eq1}) and Eq. (\ref{thm2eq2}), we can obtain the following equation.
\begin{equation}\label{eq:qaa}
\begin{split}
	&e^{i\varphi}(1-e^{i\phi})(1-2^{r+k_l+t-n})\\
 =&(1-e^{i\phi})(1-2^{r+k_l+t-n})+e^{i\phi}.
\end{split}
\end{equation}

Denote by 
\begin{equation}
b=1-2^{r+k_l+t-n}.
\end{equation}

Then according to Eq. (\ref{eq:qaa}), we have
\begin{equation}
\begin{split}
	b =& e^{-i\varphi}\left(b + \frac{1}{e^{-i\phi} - 1}\right)\\
	  =& e^{-i\varphi}\left(b + \frac{1}{\cos\phi - 1 - i\sin\phi}\right)\\
	  =& e^{-i\varphi}\left(b + \frac{\cos\phi - 1}{\left(\cos\phi - 1\right)^2 + \sin^2\phi}\right. \\
&\left.+ i\frac{\sin\phi}{\left(\cos\phi - 1\right)^2 + \sin^2\phi}\right)\\
	  =& e^{-i\varphi}\left(b - \frac{1}{2} + i\frac{\sin\phi}{2 - 2\cos\phi}\right).\label{eq:varphi}
\end{split}						  
\end{equation}

Taking the square of $b$, we can further have
\begin{equation}\label{eq:phi0}
	b^2 = \left(b - \frac{1}{2}\right)^2 + \frac{\sin^2\phi}{4\left(1 - \cos\phi\right)^2}.
\end{equation}

Arrange Eq. (\ref{eq:phi0}) to get
\begin{equation}\label{eq:phi}
\begin{split}
	4b - 1 &= \frac{\sin^2\phi}{\left(1 - \cos\phi\right)^2}\\
 &= \cot^2\frac{\phi}{2}.
\end{split} 
\end{equation}

From Eq. (\ref{eq:phi}), we can obtain
\begin{equation}
\begin{split}
	\phi  &= 2\arctan{\left(\sqrt{\frac{1}{4b-1}}\right)}\\
	      &= 2\arctan{\left(\sqrt{\frac{2^{n-t-r-k_l}}{3\cdot 2^{n-t-r-k_l}-4}}\right)}.\label{eq:phi_final}
\end{split}      
\end{equation}

Since $b$ is a real number,  Eq. (\ref{eq:varphi}) also needs to be a real number, and we can further obtain
\begin{equation}
\begin{split}
	\varphi  &= \arccos{\left(\frac{b-\frac{1}{2}}{b}\right)}\\
	         &= \arccos{\left(\frac{2^{n-t-r-k_l-1}-1}{2^{n-t-r-k_l}-1}\right)}.\label{eq:varphi_final}
\end{split}     
\end{equation}
\end{proof}

\end{document}